\documentclass[11pt]{article}

\usepackage[T1]{fontenc}
\usepackage[utf8]{inputenc}
\usepackage{lmodern}

\usepackage[margin=1in]{geometry}
\usepackage{graphicx}%
\usepackage{multirow}%
\usepackage{amsmath,amssymb,amsfonts}%
\usepackage{amsthm}%
\usepackage{mathrsfs}%
\usepackage[title]{appendix}%
\usepackage{xcolor}%
\usepackage{textcomp}%
\usepackage{manyfoot}%
\usepackage{booktabs}%
\usepackage{listings}%

\usepackage{fullpage,caption,wrapfig}
\usepackage{comment}
\usepackage{mathtools}
\usepackage[shortlabels]{enumitem}
\usepackage{complexity}
\usepackage[backend=bibtex, style=alphabetic, backref=true,maxbibnames=99, url=false]{biblatex} 
\usepackage{bbm}
\usepackage{thmtools}

\usepackage{nag,tikz}
\usetikzlibrary{calc}
\usetikzlibrary{decorations.pathreplacing}
\usetikzlibrary{shapes, patterns, decorations, fit, intersections, arrows, automata, positioning}

\usepackage{algorithm,algpseudocode}
\usepackage{multicol}
\usepackage[scaled]{helvet} 
\usepackage{thmtools} 
\usepackage{hyperref}
\hypersetup{
    colorlinks=true,
    linkcolor=violet,
    filecolor=magenta,      
    urlcolor=cyan,
    citecolor=blue,
    pdffitwindow=true,
}\usepackage[capitalize, nameinlink]{cleveref}

\theoremstyle{plain}
\newtheorem{theorem}{Theorem}[section]
\newtheorem{lemma}[theorem]{Lemma}

\newtheorem{fact}[theorem]{Fact}

\newtheorem{definition}[theorem]{Definition}

\theoremstyle{remark}

\newcommand{\old}[1]{}

\renewcommand{\R}{\ensuremath{\mathbb R}}

\renewcommand{\P}[1]{{\mathbb{P}}\left[#1\right]}
\renewcommand{\PP}[2]{{\mathbb{P}}_{#1}\left[#2\right]}

\renewcommand{\E}[1]{{\mathbb{E}}\left[#1\right]}
\renewcommand{\EE}[2]{{\mathbb{E}}_{#1}\left[#2\right]}

\renewcommand{\path}[2]{{ S_{#1}, \ldots, S_{#2} }}

\newcommand{\set}[1]{\{{#1}\}}

\def\b1{{\bf 1}}
\def\1{{\bf 1}}

\def\cT{{\cal T}}

\def\setminus{-}
\def\R{\mathbb{R}}

\newcommand{\norm}[1]{\|#1\|}

\newcommand{\bedit}[1]{{{#1}}}
\newcommand{\nedit}[1]{{{#1}}}
\newcommand{\dedit}[1]{{{#1}}}

\title{Maximum Entropy is a 10/7-Approximation Algorithm for the TSP on Half-Integral Cycle Cut Instances}
\author{Billy Jin\\Purdue University \and Nathan Klein\\Boston University \and David P.\ Williamson\\Cornell University}
\date{}

\begin{document}

\maketitle

\begin{abstract}
One of the most famous conjectures in combinatorial optimization is the four-thirds conjecture, which states that the integrality gap of the Subtour LP relaxation of the TSP is equal to $\frac43$. For 40 years, the best known upper bound was $1.5$, due to 
Wolsey \cite{Wolsey80}.  Recently, Karlin, Klein, and Oveis Gharan \cite{KKO21b} showed that the max entropy algorithm for the TSP gives an improved bound of $1.5 - 10^{-36}$.
In this paper, we show that the maximum entropy algorithm is a $\frac{10}{7}$-approximation for half-integral cycle cut instances of the TSP. This class of instances contains examples which demonstrate the subtour LP has an integrality gap of at least $\frac{4}{3}$, as well as examples showing that the performance of the max entropy algorithm is no better than $\frac{11}{8}$. 
\nedit{We note that in \cite{JinKW23}, the authors gave an algorithm upper bounding the integrality gap of this class of instances by $\frac{4}{3}$, so this work does not (and could not) provide an improved bound on the integrality gap. However, since} there is no reason to believe that the analysis of the maximum entropy algorithm on general instances is tight, our work provides hope (and potentially direction) for improved analysis on other instance classes.
\end{abstract} 

\section{Introduction}

In the traveling salesman problem (TSP), we are given a set of $n$ cities and the costs $c_{ij}$ of traveling from city $i$ to city $j$  for all $i,j$. The goal of the problem is to find the cheapest tour that visits each city exactly once and returns to its starting point.
An instance of the TSP is called {\em symmetric} if $c_{ij} = c_{ji}$ for all $i,j$; it is {\em asymmetric} otherwise.  Costs obey the {\em triangle inequality} (or are {\em metric}) if $c_{ij} \leq c_{ik} + c_{kj}$ for all $i,j,k$.  All instances we consider will be symmetric and obey the triangle inequality. We treat the problem input as a complete graph $G=(V,E)$, where $V$ is the set of cities, and $c_e = c_{ij}$ for edge $e=\{i,j\}$.

In the mid-1970s, Christofides \cite{Christofides76} and Serdyukov \cite{Serdyukov78} each gave a
$\frac{3}{2}$-approximation algorithm for the symmetric TSP with
triangle inequality.  The algorithm computes a minimum-cost spanning tree and then finds a minimum-cost perfect matching on the odd degree vertices of the tree to compute a connected Eulerian subgraph. Because the edge costs satisfy the triangle inequality, any Eulerian tour of this Eulerian subgraph can be ``shortcut'' to a tour of no greater cost.  Until very recently, this was the best approximation factor known for the symmetric TSP with triangle inequality, although over the last decade substantial progress was made for many special cases and variants of the problem. 

In recent years, a variation on the Christofides-Serdyukov algorithm has been considered. Its starting point is  a well-known linear programming relaxation of the TSP introduced by Dantzig, Fulkerson, and Johnson \cite{DFJ54}, sometimes called the {\em Subtour LP} or the {\em Held-Karp bound} \cite{HeldK71}.  The Subtour LP is as follows:
\begin{equation}
\begin{aligned}
    \min \quad & \sum_{e \in E} c_e x_e \\
    \text{s.t.} \quad & x(\delta(v)) = 2, &\forall \; v \in V, \\
    & x(\delta(S)) \geq 2, & \forall \; S\subset V, S \neq \emptyset, \\
    & 0 \leq x_e \leq 1, &\forall e \in E,
\end{aligned}
\label{LP}
\end{equation}
where $\delta(S)$ is the set of all edges with exactly one endpoint in $S$ and we use the shorthand that $x(F) = \sum_{e \in F} x_e$.  \dedit{Wolsey \cite{Wolsey80} shows that the minimum-cost spanning tree is at most the value of the Subtour LP, and a matching on its odd degree vertices is at most half the value of the Subtour LP, showing that the Christofides-Serdyukov algorithm has cost at most $\frac{3}{2}$ the Subtour LP.}  \dedit{Following Wolsey, } it is not difficult to show that for any  solution $x^*$ of this LP relaxation, $\frac{n-1}{n}x^*$ is a feasible point in the spanning tree polytope, i.e., the convex hull of all spanning trees of the graph.  Therefore, $\frac{n-1}{n}x^*$ can be decomposed into a convex combination of spanning trees, and the cost of this convex combination is a lower bound on the cost of an optimal tour. Any such convex combination 
can be viewed as a distribution over spanning trees such that the expected cost of a spanning tree sampled from this distribution is a lower bound on the cost of an optimal tour; note that there can be many possible such convex combinations, giving rise to many possible distributions.  The variation of the Christofides-Serdyukov algorithm considered is one that samples a random spanning tree from some distribution arising from a convex combination, 
and then finds a minimum-cost perfect matching on the odd vertices of the tree.  This idea was introduced in work of Asadpour, Goemans, M\k{a}dry, Oveis Gharan, and Saberi \cite{AsadpourGMOS17} (in the context of the asymmetric TSP) and Oveis Gharan, Saberi, and Singh \cite{OveisGharanSS11} (for symmetric TSP).

Asadpour et al.\ \cite{AsadpourGMOS17} and Oveis Gharan, Saberi, and Singh \cite{OveisGharanSS11} consider a particular distribution of spanning trees known as the {\em maximum entropy distribution}.  The maximum entropy algorithm finds a probability distribution $p_T$ on spanning trees $T$ such that the marginal distribution on each edge $e$ is $\frac{n-1}{n}x^*_e$ (that is, $\sum_{T:e\in T} p_T = \frac{n-1}{n} x^*_e$) and that maximizes the entropy function $-\sum_T p_T \log p_T$.    We will call the algorithm that samples from the maximum entropy distribution and then finds a minimum-cost perfect matching on the odd degree vertices of the tree the {\em maximum entropy algorithm} for the symmetric TSP.  

In a breakthrough result, Karlin, Klein, and Oveis Gharan \cite{KKO21} show that a variant of the maximum entropy algorithm has performance ratio better than 3/2, although the amount by which the bound was improved is quite small (approximately $10^{-36}$; the $\epsilon$ was then increased to $10^{-34}$ in \cite{GKL24}). The achievement of the paper is to show that choosing a random spanning tree from the maximum entropy distribution gives a distribution of odd degree nodes in the spanning tree such that the expected cost of the perfect matching is cheaper (if marginally so) than in the \dedit{Wolsey} analysis.  Note that the Karlin et al.\ algorithm actually samples from the set of {\em 1-trees}, a spanning tree plus one additional edge\footnote{Held and Karp \cite{HeldK71} define a 1-tree to be a tree with two distinct edges incident on a specific vertex plus a spanning tree on the remaining vertices. Our definition is more general.}, and then finds a matching on the odd-degree vertices. In this paper, we will study this algorithm based on 1-trees.

\subsection{Our Contribution}
We show that the maximum entropy algorithm studied in \cite{KKO21} is a randomized $\frac{10}{7}$-approximation algorithm for a class of TSP instances known as {\em half-integral cycle cut} instances, which are defined as follows.

\begin{definition} 
\nedit{A pair} $(G,x)$ is a \textbf{half-integral} TSP instance if \nedit{$G$ is an instance of TSP,} $x$ is a feasible solution to the Subtour LP, and $x$ satisfies $x_e \in \{0, \frac{1}{2}, 1\}$ for all $e \in E$. It is a {\textbf{cycle cut}} instance if for any set $S$ with $x(\delta(S)) = 2$ (we call such a set $S$ a \textbf{tight set}) and $|S| \geq 2$, there is a partition $A\cup B$ of $S$ such that both $A$ and $B$ are tight. 
\end{definition}
\nedit{Not all half-integral instances are cycle cuts. For example, $K_5$ with $x_e = \frac{1}{2}$ on every edge is not a cycle cut instance. We will show the following theorem:
\begin{restatable}{theorem}{mainthm}\label{thm:main}
Let $(G,x)$ be a half-integral cycle cut instance with $G=(V,E)$. Let $\mu$ be the max entropy distribution over 1-trees with marginals $x$, for any root $r \in V$. For $T \sim \mu$, let $M$ be the cheapest perfect matching on the odd degree vertices of $T$. Then, 
$$\EE{T \sim \mu}{c(T) + c(M)} \le \frac{10}{7}c(x).$$
In other words, the max entropy algorithm is a $\frac{10}{7}$ approximation for $(G,x)$.
\end{restatable}
}
Note that this implies a $\frac{10}{7}$-approximation algorithm for instances for which we can find an optimal half-integral cycle cut solution $x$. 

One reason why half-integral instances are interesting is a conjecture by Schalekamp, Williamson, and van Zuylen \cite{SchalekampWvZ14}, which states that these are the worst-case instances for the integrality gap of the subtour LP. Currently, the best known approximation ratio for half-integral TSP is 1.49776, due to Klein and Taziki \cite{KT25}, which analyzes the max entropy algorithm, \nedit{and builds upon other recent work on half integral TSP \cite{Gupta2024,HaddadanN19,KKO20}.} Cycle cut instances are also interesting because recent papers that break the $\frac32$ bound for TSP decompose the graph into a hierarchy of cuts, where each cut in the hierarchy can be either a \emph{cycle cut} or a \emph{degree cut}. Our definition of cycle cut instances corresponds to the case where all cuts in the hierarchy are cycle cuts, and so it is plausible that studying this case could provide insights about the more general case. \nedit{The class of half-integral cycle cuts contains instances on which the integrality gap is at least $\frac{4}{3}$, including the well-known envelope graph.}

The authors have previously shown that there is a $\frac{4}{3}$-approximation algorithm \nedit{and an upper bound of $\frac{4}{3}$ on the integrality gap} for half-integral cycle cut instances  \cite{JinKW25}, and that there exists a half-integral cycle cut instance on which the maximum entropy algorithm produces a solution with expected cost at least $\frac{11}{8}-o(1)$ times the optimum \cite{JinKW25b}. The latter result shows that the maximum entropy algorithm is \emph{not} a conjectured $\frac{4}{3}$-approximation algorithm for TSP. Still, our result is the first to show that maximum entropy algorithm \nedit{has an approximation factor of better than 1.49} on a non-trivial class of instances. Moreover, the $\frac43$-approximation algorithm in \cite{JinKW25} is specialized to half-integral cycle cut instances, and it is not clear how to extend it to general TSP instances. Since there is no reason to believe that the analysis of the maximum entropy algorithm on general instances is tight, our work provides hope (and potentially direction) for improved analysis on other instance classes. \dedit{The analysis of \cite{JinKW25} uses the stationary point of a Markov chain for its argument.  This paper uses ideas from that analysis but without the Markov chain, instead reasoning directly about the stationary distribution.} \nedit{We find it promising that max entropy, while not able to obtain the $\frac{4}{3}$ ratio of \cite{JinKW25} (which is specialized to these instances), is nonetheless able to simulate that algorithm with a relatively small loss in approximation ratio.} It remains an interesting open question to close the gap in the performance guarantee of the maximum entropy algorithm between the lower bound of $\frac{11}{8}$ and the upper bound of $\frac{10}{7}$ on half-integral cycle cut instances.

\section{Preliminaries and the Algorithm}

We split every edge $e$ with $x_e = 1$ into two edges of value $\frac{1}{2}$ so that $G$ is a 4-regular multigraph with $x_e =\frac{1}{2}$ on every edge. Note that for any tight set $S$ (for which $x(\delta(S))=2$), there are four edges in $\delta(S)$. We construct a hierarchy of tight sets as follows.\footnote{\nedit{We remark that the hierarchy here is similar but slightly different than that considered in prior work, e.g., \cite{KKO20,JinKW23}. Our hierarchy is binary on cycle cuts, which leads to a simpler analysis, whereas prior hierarchies (even in the cycle cut case) would have sets with multiple children.}} We begin by selecting a root vertex $r$ arbitrarily. Note that both $\{r\}$ and $V\setminus \set{r}$ are tight cuts. Using the definition of a cycle cut instance, we can divide $V\setminus \set{r}$ into two tight sets, and then recursively divide each of these into two tight sets, yielding a binary tree ${\cal T}$ of tight sets. For a tight set $S$ in the tree partitioned into two tight sets $A$ and $B$, we observe that of the four edges in $\delta(S)$ (with $x_e= 1/2$), exactly two are also in $\delta(A)$ and the remaining two are also in $\delta(B)$; thus there are two edges in $\delta(A) \cap \delta(B)$. See Figures \ref{fig:root} and \ref{fig:notroot}. Note that each leaf of the tree corresponds to a vertex of the graph. 

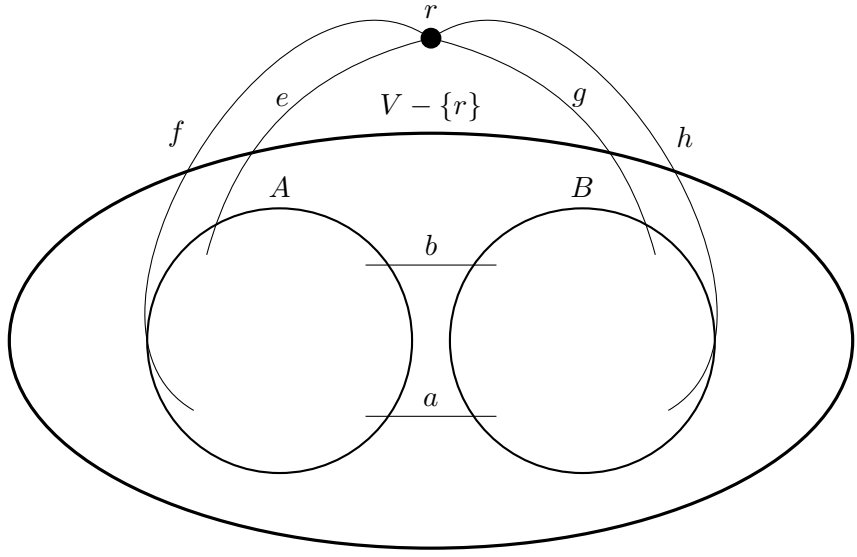
\begin{figure}
    	\begin{center}
		\begin{tikzpicture}[point/.style={circle,draw=black,thick,inner sep=0,minimum size=1mm,fill=black}]
			
			\node (1) at (0,0) {};
			\node (2) at (2,0) {};
			\node (3) at (1,1) {};
			\node (4) at (0,2) {};
			\node (5) at (2,2) {};
			
			\node (6) at (4,0) {};
			\node (7) at (6,0) {};
			\node (8) at (5,1) {};
			\node (9) at (4,2) {};
			\node (10) at (6,2) {};
			
			\node[circle, draw=black, minimum size=3.5cm, thick, label=$A$] (A) at (3) {};
			\node[circle, draw=black, minimum size=3.5cm, thick, label=$B$] (B) at (8) {};
			
			\node[ellipse, draw, fit=(1) (5) (7) (10) (A) (B), inner sep=5pt, label=$V-\{r\}$, very thick] {};
			
			\node[point, label=$r$] (r) at (3,5) {$r$};
			
			\path (2) edge node[above] {$a$} (6)
			      (5) edge node[above] {$b$} (9);
			      
			\path (r) edge[bend right] node[above] {$e$} (4);
			\draw (r) to [out=150, in=150] node[above left] {$f$} (1);
			\path (r) edge[bend left] node[above] {$g$} (10);
			\draw (r) to [out=30, in=30] node[above right] {$h$} (7);
			
		\end{tikzpicture}
	\end{center}
 
    \caption{Edges and tight sets at the root node $r.$}   \label{fig:root}
\end{figure}

\begin{figure}
	\begin{center}
		\begin{tikzpicture}[point/.style={circle,draw=black,thick,inner sep=0,minimum size=1mm,fill=black}]
			
			\node (1) at (0,0) {};
			\node (2) at (2,0) {};
			\node (3) at (1,1) {};
			\node (4) at (0,2) {};
			\node (5) at (2,2) {};
			
			\node (6) at (4,0) {};
			\node (7) at (6,0) {};
			\node (8) at (5,1) {};
			\node (9) at (4,2) {};
			\node (10) at (6,2) {};
			
			\node[circle, draw=black, minimum size=3.5cm, thick, label=$A$] (A) at (3) {};
			\node[circle, draw=black, minimum size=3.5cm, label=$B$, thick] (B) at (8) {};
			
			\node[ellipse, draw, fit=(1) (5) (7) (10) (A) (B), inner sep=5pt, very thick, label=$S$] {};
			
			\node (e) at (-1,5) {};
			\node (f) at (-2,5) {};
			\node (g) at (7,5) {};
			\node (h) at (8,5) {};
			
			\path (2) edge node[above] {$a$} (6)
			(5) edge node[ above] {$b$} (9)
			(e) edge node[near start,above right] {$e$} (4)
			(f) edge node[near start,left] {$f$} (1)
			(g) edge node[near start,above left] {$g$} (10)
			(h) edge node[near start,right] {$h$} (7);

		\end{tikzpicture}
	\end{center}

\caption{Edges and tight sets at non-root node $S$.}\label{fig:notroot}
\end{figure}
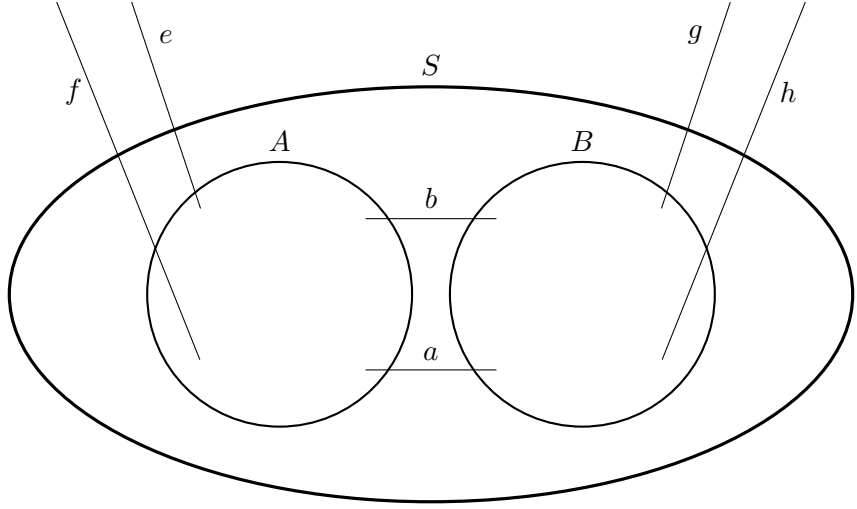

\nedit{On cycle cut instances, max entropy coincides with a very natural sampling procedure: at each level, choose one edge between the two children independently, sampling each edge $e$ with probability $x_e$.} We now describe the maximum entropy algorithm on half-integral cycle cut instances more formally.\footnote{See Jin, Klein, and Williamson \cite{JinKW25b} for a full description of how the algorithm works on general TSP instances, as well as a discussion of variants of the maximum entropy algorithm and why 
the algorithm studied behaves identically to the variant studied in \cite{KKO21} on half integral cycle cut instances.}
If $V-\set{r}$ is partitioned into two tight sets $A$ and $B$ with edges $\set{e,f} = \delta(r) \cap \delta(A)$ and $\set{g,h} = \delta(r) \cap \delta(B)$, then the maximum entropy algorithm selects exactly one edge from $\{e,f\}$ and one from $\{g,h\}$ independently and uniformly at random.  For the two edges $\{a,b\} = \delta(A) \cap \delta(B)$ joining any pair of sets in the tree ${\cal T}$ that are children of a tight set $S$ in the tree, the maximum entropy algorithm selects exactly one of $\{a,b\}$ uniformly at random. \nedit{We emphasize that the choices made at each level of ${\cal T}$ are independent of one another.} We call the two edges $a,b$ such that $\{a,b\} = \delta(A) \cap \delta(B)$ \emph{cycle partners} for sibling sets $A, B$ in the tree ${\cal T}$.

\nedit{It may be instructive to notice why the resulting set of edges is connected. This has a quick inductive proof: at the bottom of the hierarchy, the vertices are connected components. Now, if we inductively assume the two children are connected components, then by adding an edge between them, the parent becomes a connected component as well.}

Let $\nu$ be the distribution of edges given by the maximum entropy algorithm.  In what follows, we show how to extend this distribution to a distribution over connected Eulerian graphs where the expected cost of the graph is at most $\frac{10}{7}$ times the LP value. In particular, we will find a distribution $\mu$ over spanning Eulerian multi-subgraphs which stochastically dominates $\nu$. In other words, to produce $\mu$, we will first sample from $\nu$ and then add edges. 
This will imply our result, since then $\mu$ can be written as a distribution which first samples a tree $T$ from $\nu$ and then adds a matching on the odd degree vertices of $T$, and the max entropy algorithm adds the cheapest such matching.


\section{Analysis}

As discussed, here we will construct a distribution $\mu: \{0,1,2\}^E \to \R_{\ge 0}$ over spanning Eulerian multi-subgraphs which arises from sampling a tree from $\nu$ and then adding edges. For each edge $e \in E$, we will let $Y_e \in \{0,1,2\}$ be the random variable counting the number of copies of $e$ sampled by $\mu$. Let $X_e \in \{0,1\}$ denote the parity of $Y_e$, so that $X_e=1$ if $Y_e=1$ and $X_e = 0$ otherwise.  
Given a tight set $S \subseteq V$ with $|\delta(S)|=4$, let $X(\delta(S))$ denote $\sum_{e \in \delta(S)} X_e$. That is, $X(\delta(S))$ is the number of edges in $\delta(S)$ that have odd parity.   
\begin{definition}
\label{def:stationary}
For a distribution $\mu$ on edges and a tight cut $S$, $\mu$ is \textbf{stationary} on $\delta(S)$ if
\begin{enumerate}[(i)]
	\item $\P{X_e=X_f=1, X(\delta(S))=2} = \frac{1}{7}$ for all $e,f \in \delta(S)$ with $e \not= f$.
	\item $\P{X(\delta(S))=4} = \frac{1}{14}$.
	\item $\P{X(\delta(S))=0} = \frac{1}{14}$. 
\end{enumerate}
Note that this saturates the probability function, so $X(\delta(S))\in \{0,2,4\}$ with probability 1. 
\end{definition}
In other words, $\mu$ is stationary on $\delta(S)$ if with probability $\frac{1}{14}$ each edge in $\delta(S)$ is used an odd number of times,\footnote{\nedit{As we never use edges three times, this is equivalent to saying every edge is present with one copy.}} with probability $\frac{1}{14}$ each edge in $\delta(S)$ is used an even number of times, and with probability $\frac{6}{7}$, exactly two edges from $\delta(S)$ are used an odd number of times, with this pair chosen uniformly at random among the 6 possible pairs. \nedit{We remark that we use the term stationary following \cite{JinKW25}, which found a stationary distribution on a Markov chain of states encoding the parities of each vertex. Our approach here is inspired by that work.}

Below we show that if $\mu$ has three desirable properties, it gives a $\frac{10}{7}$ approximation. We will then show we can construct a distribution $\mu$ with these properties.
\begin{lemma}\label{lem:dist}
If there is a distribution $\mu: \{0,1,2\}^E \to [0,1]$ such that:
\begin{enumerate}[(1)]
	\item 
    For every $S\in \cT$, $\mu$ is stationary on $\delta(S)$.
	\item For every $S \in \cT$, $\P{Y(\delta(S))\ge 2}=1$, and for every pair of cycle partners $a,b$ of siblings $A, B \in \cT$, $\PP{Y \sim \mu}{Y_a+Y_b \ge 1} = 1$.  
	\item For every edge $e$, $\P{Y_e=2} = \frac{3}{28}$,
\end{enumerate}
then $\mu$ is a distribution over connected Eulerian multi-subgraphs with $\EE{Y \sim \mu}{Y_e} = \frac{5}{7}$ for all $e \in E$. 
Sampling from such a distribution $\mu$ gives a randomized $\frac{10}{7}$-approximation algorithm since the expected number of times each edge is used is $\frac{10}{7}x_e$. \end{lemma}
\begin{proof}
	Any subgraph sampled from $\mu$ is connected due to property (2). 
    Every set in $\cT$ has even parity with probability 1 due to property (1). This implies that every vertex in $V \smallsetminus \{r\}$ has even parity, because they are in $\cT$. Also, $r$ has even parity since $V \smallsetminus \{r\} \in \cT$. 
    Therefore, $\mu$ is a  distribution over spanning Eulerian multi-subgraphs.
	
	It remains to show that $\E{Y_e} = \frac{10}{7}x_e = \frac{5}{7}$. Fix $e \in E$ such that $e \in \delta(v)$. Since $\mu$ is stationary on $\delta(v)$, we have that $\P{X_e=1} = \frac{3}{7} + \frac{1}{14} = \frac{1}{2}$ using (i) and (ii) of \Cref{def:stationary}. Furthermore, from property (3) we have that $\P{Y_e=2} = \frac{3}{28}$. So, $\E{Y_e} = \frac{1}{2} + 2 \cdot \frac{3}{28} = \frac{5}{7}$, as desired.
\end{proof}
Let $\nu$ be the max entropy distribution over 1-trees with marginals $x$. We will now extend $\nu$ to a distribution $\mu$ which has the properties listed in \Cref{lem:dist}. In particular, $\mu$ will extend each sampled tree $T$ by adding an $\mathrm{odd}(T)$-join.\footnote{\bedit{Here, an $O$-join refers to a set of edges whose odd-degree vertices is exactly $O$, and $\mathrm{odd}(T)$ is the set of odd-degree vertices in $T$. So an $\mathrm{odd}(T)$-join is exactly what needs to be added to $T$ to make the parity of every vertex even.}} 

    

In the following lemma, we will extend $\nu$ only on the edges adjacent to the root, forming the base case. In \cref{lem:inductiveStep} we will then show how to continue this inductively.
\begin{lemma}\label{lem:basecase}
There is an extension $\mu$ of $\nu$ that is stationary on $\delta(r)$, satisfies $Y(\delta(r)) \ge 2$ with probability 1, and satisfies (3) in \Cref{lem:dist} for all $e \in \delta(r)$.
\end{lemma}
\begin{proof}
As described previously, the maximum entropy algorithm splits $\delta(r)$ into two groups of edges $\{e,f\}$ and $\{g,h\}$ and picks exactly one edge from $\{e,f\}$ and exactly one from $\{g,h\}$ independently; see Figure \ref{fig:root}. To extend this distribution, after sampling edges $\{p,q\}$ where $p \in \{e,f\}$ and $q \in \{g,h\}$:
\begin{enumerate}
	\item With probability $\frac{1}{14}$, add the other two edges.
	\item With probability $\frac{1}{14}$, double both edges. 
	\item With probability $\frac{1}{7}$, add the other edge in $\{e,f\}$ and double $q$.
	\item With probability $\frac{1}{7}$, add the other edge in $\{g,h\}$ and double $p$.
    \item \bedit{With probability $\frac47$, do nothing.}
\end{enumerate}  	
Clearly, $Y(\delta(r)) \ge 2$ with probability 1. Second, for any given edge (\nedit{say in $\{e,f\}$}), it is selected by the maximum entropy tree with probability $\frac{1}{2}$. If selected, it is doubled if case 2 or 4 occurs, for a probability of $\frac{1}{2} (\frac{1}{14} + \frac{1}{7}) = \frac{3}{28}$. If not selected, it is never doubled. Thus any edge in $\delta(r)$ is doubled with probability $\frac{3}{28}$.

It remains to show that the distribution is stationary. Properties (ii) and (iii) in \Cref{def:stationary} follow immediately from 1. and 2., \bedit{respectively}. To verify property (i), we show that any pair of edges ends up with odd parity while the other two have even parity with probability exactly $1/7$.
\begin{itemize}
    \item For a ``crossing'' pair (e.g., $e$ and $g$), they are selected by the max entropy algorithm with probability $1/4$. The outcome is $\{X_e=1, X_g=1, X(\delta(r))=2\}$ if no modification occurs, which happens with probability $1 - (\frac{1}{14} + \frac{1}{14} + \frac{1}{7} + \frac{1}{7}) = \frac{4}{7}$. The total probability is $\frac{1}{4} \cdot \frac{4}{7} = \frac{1}{7}$.
    \item For a ``parallel'' pair (e.g., $e$ and $f$), the outcome $\{X_e=1, X_f=1, X(\delta(r))=2\}$ occurs if we sample $\{e,q\}$ (for $q \in \{g,h\}$) and apply rule 3, or if we sample $\{f,q\}$ and apply rule 3. Each occurs with probability $\frac{1}{2} \cdot \frac{1}{7} = \frac{1}{14}$. Summing these disjoint events, the total probability is $\frac{1}{14}+\frac{1}{14} = \frac{1}{7}$.
\end{itemize}
Thus property (i) of \Cref{def:stationary} holds.
\end{proof}
The following lemma provides the inductive step. Given that we have already processed a cut $S$ to ensure $\delta(S)$ is stationary, we show how to process its children $A$ and $B$. This ensures that by the time we have processed all cuts in $\mathcal{T}$, the properties of \Cref{lem:dist} hold for all cuts and edges. As observed previously, for any $S \in \cT$ which is not a singleton vertex, $S$ has two children $A$ and $B$, there are two edges between $A$ and $B$, and $|\delta(A)\cap\delta(S)|=|\delta(B)\cap\delta(S)|=|\delta(A) \cap \delta(B)| = 2$. See \Cref{fig:notroot} for an illustration. 
\begin{lemma}\label{lem:inductiveStep}
    Let $S \in \cT$ with children $A$ and $B$. Suppose $\mu$ is stationary over $\delta(S)$. If $|S| \geq 2$, we can extend $\mu$ so that \bedit{property (1) of \Cref{lem:dist} holds on $A$ and $B$,} and  properties (2) and (3) of \Cref{lem:dist} hold for the cuts $A$ and $B$ and for the two edges between them.
    \end{lemma}
\begin{proof}
Let $A,B$ be the children of $S$. Let $\{e,f,g,h\}$ be the edges in $\delta(S)$ where $e,f \in \delta(A)$ and $g,h \in \delta(B)$. Let $a,b$ be the cycle partners between $A$ and $B$. See Figure \ref{fig:notroot}. Now we perform the following extension of $\mu$ to include $a,b$ so that $\delta(A)$ and $\delta(B)$ are stationary, we always take at least one of $a$ or $b$, and the probability we double each of $a,b$ is $\frac{3}{28}$. (Note that unlike in the root case, it is not necessarily the case that max entropy takes exactly one edge among each of $e,f$ and $g,h$.) 
\begin{enumerate}
	\item If $X(\delta(S))=2$ and $X_e+X_f=1$ and $X_g+X_h=1$, do nothing. Max entropy will take $a$ or $b$ uniformly at random.
	\item When $\{X_e=X_f=1, \, X(\delta(S))=2\}$ or $\{X_g=X_h=1, \, X(\delta(S))=2\}$, with probability $\frac{1}{2}$ double the edge among $a,b$ picked by max entropy. \bedit{Otherwise, add the other edge.}
	\item When $X(\delta(S)) = 4$, double the edge among $a,b$ that max entropy picks.
	\item When $X(\delta(S)) = 0$, add the edge among $a,b$ that max entropy does not pick.
\end{enumerate}
Max entropy takes $a$ or $b$, and we do not delete edges, so property (2) of \Cref{lem:dist} is satisfied. The probability we double an edge is $\frac{2}{7}\cdot \frac{1}{4}$ in case 2 and $\frac{1}{14}\cdot \frac{1}{2}$ in case 3. Summing these we get $\frac{3}{28}$ as desired. So it remains to show that this distribution is stationary for $\delta(A)$ and $\delta(B)$.

For $\delta(A)$, $X(\delta(A)) = 4$ occurs only when we take the pair $e,f$ (w.p. $\frac{1}{7}$) and then we take both $a$ and $b$ (w.p. $\frac{1}{2}$). This gives a probability of $\frac{1}{7} \cdot \frac{1}{2} = \frac{1}{14}$ as desired. $X(\delta(A)) = 0$ occurs only when we take the pair $g,h$ (w.p. $\frac{1}{7}$) and we take $a$ twice or we take $b$ twice (w.p. $\frac{1}{2}$), so this happens with probability $\frac{1}{7} \cdot \frac{1}{2} = \frac{1}{14}$. 

Next, the event $\{X_e=1, X_f = 1, X(\delta(A))=2\}$ happens when either $X(\delta(S))=4$ (w.p. $\frac{1}{14}$), or $X_e=X_f=1$ (w.p. $\frac{1}{7}$) and we double $a$ or we double $b$ (w.p. $\frac{1}{2}$). So this occurs with probability $\frac{1}{7} \cdot \frac{1}{2} + \frac{1}{14} = \frac{1}{7}$. On the other hand, the event $\{X_a=1, X_b=1, X(\delta(A))=2\}$ happens when either $X(\delta(S))=0$ (w.p. $\frac{1}{14}$), or $X_g=X_h=1$  (w.p. $\frac17$) and we take both $a$ and $b$ (w.p. $\frac{1}{2}$).  So, this happens with probability $\frac{1}{14} + \frac{1}{7}\cdot \frac{1}{2} = \frac{1}{7}$. Finally, we consider a pair that consists of one edge of $e,f$ (say $e$) and one edge of $a,b$ (say $a$). In this case, the event $\{X_e=X_a=1,X(\delta(A))=2\}$ happens when $X_e=1$ and $X_f=0$ in $\delta(S)$, which occurs with probability $\frac{2}{7}$, and then we select $a$, which occurs with probability $\frac{1}{2}$. So the total probability of this event is also $\frac{1}{7}$, as desired.

This completes the proof, as an identical argument works for $\delta(B)$.
\end{proof}

\nedit{
The main theorem is now a corollary of the preceding three lemmas. In summary, \cref{lem:basecase} and \cref{lem:inductiveStep} show that given $\nu$, the max entropy distribution, one can add edges to elements in the support of $\nu$ (i.e., extend $\nu$) to produce a distribution $\mu$ with the properties described in \cref{lem:dist}. The expected cost of an element sampled from $\mu$ is at most $\frac{10}{7}c(x)$. Since the max entropy algorithm adds the cheapest collection of edges making the tree Eulerian, the expected cost of the max entropy algorithm is also at most $\frac{10}{7}c(x)$. \\
}

\noindent \textbf{Acknowledgments} \\

\noindent The second author was supported by the NSF CAREER grant CCF-2442250.

\printbibliography

@string{focs11 = "Proceedings of the 52nd Annual {IEEE} Symposium on the Foundations of Computer Science"}

@string{stoc20 = "Proceedings of the 52nd Annual {ACM} Symposium on the
	Theory of Computing"}

@string{a = "Algorithmica"}

@string{c = "Combinatorica"}

@string{mor = "Mathematics of Operations Research"}

@string{mpb = "Mathematical Programming B"}

@string{or = "Operations Research"}

@article{AsadpourGMOS17,
author = "Arash Asadpour and Michel X. Goemans and Aleksander Mądry and Shayan {Oveis Gharan} and Amin Saberi",
title	= "An ${O}(\log n/\log \log n)$-Approximation Algorithm for the Asymmetric Traveling Salesman Problem",
journal = or,
volume = 65,
pages	= "1043--1061",
year	= 2017}

@InProceedings{GKL24,
  author =	{Gurvits, Leonid and Klein, Nathan and Leake, Jonathan},
  title =	{{From Trees to Polynomials and Back Again: New Capacity Bounds with Applications to TSP}},
  booktitle =	{51st International Colloquium on Automata, Languages, and Programming (ICALP 2024)},
  pages =	{79:1--79:20},
  series =	{Leibniz International Proceedings in Informatics (LIPIcs)},
  ISBN =	{978-3-95977-322-5},
  ISSN =	{1868-8969},
  year =	{2024},
  volume =	{297},
  editor =	{Bringmann, Karl and Grohe, Martin and Puppis, Gabriele and Svensson, Ola},
  publisher =	{Schloss Dagstuhl -- Leibniz-Zentrum f{\"u}r Informatik},
  address =	{Dagstuhl, Germany},
  URL =		{https://drops.dagstuhl.de/entities/document/10.4230/LIPIcs.ICALP.2024.79},
  URN =		{urn:nbn:de:0030-drops-202229},
  doi =		{10.4230/LIPIcs.ICALP.2024.79}
}

@techreport{Christofides76,
author  = "Nicos Christofides",
title   = "Worst Case Analysis of a New Heuristic for the Traveling Salesman
       Problem",
institution = "Graduate School of Industrial Administration,
           Carnegie-Mellon University",
type    = "Report",
number  = 388,
address = "Pittsburgh, PA",
year    = 1976}

@InProceedings{HaddadanN19,
  author =	{Arash Haddadan and Alantha Newman},
  title =	{{Towards Improving Christofides Algorithm for Half-Integer {TSP}}},
  booktitle =	{27th Annual European Symposium on Algorithms (ESA 2019)},
  pages =	{56:1--56:12},
  series =	{Leibniz International Proceedings in Informatics (LIPIcs)},
  ISBN =	{978-3-95977-124-5},
  ISSN =	{1868-8969},
  year =	{2019},
  volume =	{144},
  editor =	{Michael A. Bender and Ola Svensson and Grzegorz Herman},
  publisher =	{Schloss Dagstuhl--Leibniz-Zentrum fuer Informatik},
  address =	{Dagstuhl, Germany},
  note="See also \url{https://arxiv.org/pdf/1907.02120v3.pdf}"}

@article{HeldK71,
author  = "Michael Held and Richard M. Karp",
title   = "The Traveling-Salesman Problem and Minimum Spanning Trees",
journal = or,
volume  = 18,
pages   = "1138--1162",
year    = 1971}

@inproceedings{JinKW23,
author="Billy Jin and Nathan Klein and  David P. Williamson",
editor="Del Pia, Alberto and Kaibel, Volker",
title="A 4/3-Approximation Algorithm for Half-Integral Cycle Cut Instances of the {TSP}",
booktitle="Integer Programming and Combinatorial Optimization",
series="Lecture Notes in Computer Science",
number=13904,
publisher="Springer International Publishing",
address="Berlin, Germany",
pages="217--230",
year=2023
}

@article{JinKW25,
author="Billy Jin and Nathan Klein and David P. Williamson",
title="A 4/3-Approximation Algorithm for Half-Integral Cycle Cut Instances of the {TSP}",
journal = mpb,
volume = 210,
pages = "511--538",
year = 2025
}

@article{JinKW25b,
author="Billy Jin and Nathan Klein and David P. Williamson",
title="A Lower Bound for the Maximum Entropy Algorithm for {TSP}",
journal=mpb,
year = 2025,
note="To appear.  See also \url{https://arxiv.org/abs/2311.01950}."}

@inproceedings{OveisGharanSS11,
author  = "Shayan {Oveis Gharan} and Amin Saberi and Mohit Singh",
title   = "A Randomized Rounding Approach to the Traveling Salesman Problem",
booktitle = focs11,
pages   = "550--559",
year    = 2011}

@article{SchalekampWvZ14,
  shorthand = {SWvZ14},
  author    = {Frans Schalekamp and
               David P. Williamson and
               Anke {van Zuylen}},
  title     = {2-Matchings, the Traveling Salesman Problem, and the Subtour {LP:}
               {A} Proof of the {B}oyd-{C}arr Conjecture},
  journal   = mor,
  volume    = {39},
  number    = {2},
  pages     = {403--417},
  year      = {2014}}

@article{Serdyukov78,
author	= "A. Serdyukov",
title	= "On Some Extremal Walks in Graphs",
journal	= "Upravlyaemye Sistemy",
volume	= 17,
pages	= "76--79",
year	= 1978}

@article{Wolsey80,
author  = "L. A. Wolsey",
title   = "Heuristic Analysis, Linear Programming and Branch and Bound",
journal = "Mathematical Programming Study",
volume  = 13,
pages   = "121--134",
year    = 1980}

@inproceedings{KKO20,
  author    = {Anna R. Karlin and
               Nathan Klein and
               Shayan {Oveis Gharan}},
  title     = {An Improved Approximation Algorithm for {TSP} in the Half Integral
            Case},
  booktitle = stoc20,
  pages     = {28--39},
  year      = {2020},
}

@INPROCEEDINGS{KKO21b,
author = {A. Karlin and N. Klein and S. {Oveis Gharan}},
booktitle = {2022 IEEE 63rd Annual Symposium on Foundations of Computer Science (FOCS)},
title = {A (Slightly) Improved Bound on the Integrality Gap of the Subtour {LP} for {TSP}},
year = {2022},
volume = {},
issn = {},
pages = {832-843},
url = {https://doi.ieeecomputersociety.org/10.1109/FOCS54457.2022.00084},
publisher = {IEEE Computer Society},
address = {Los Alamitos, CA, USA},
month = {11}
}

@article{KKO21,
author = {Karlin, Anna R. and Klein, Nathan and {Oveis Gharan}, Shayan},
title = {A (Slightly) Improved Approximation Algorithm for Metric TSP},
journal = {Operations Research},
volume = {72},
number = {6},
pages = {2543-2594},
year = {2024},
doi = {10.1287/opre.2022.2338},

URL = { 
    
        https://doi.org/10.1287/opre.2022.2338
    
    

},
eprint = { 
    
        https://doi.org/10.1287/opre.2022.2338
    
    

}
}

@Article{Gupta2024,
author={Gupta, Anupam
and Lee, Euiwoong
and Li, Jason
and Mucha, Marcin
and Newman, Heather
and Sarkar, Sherry},
title={Matroid-based TSP rounding for half-integral solutions},
journal={Mathematical Programming},
year={2024},
month={Jul},
day={01},
volume={206},
number={1},
pages={541-576},
issn={1436-4646},
doi={10.1007/s10107-024-02065-4},
url={https://doi.org/10.1007/s10107-024-02065-4}
}

@article{DFJ54,
 URL = {http://www.jstor.org/stable/166695},
 author = {G. Dantzig and R. Fulkerson and S. Johnson},
 journal = {Journal of the Operations Research Society of America},
 number = {4},
 pages = {393--410},
 publisher = {INFORMS},
 title = {Solution of a Large-Scale Traveling-Salesman Problem},
 urldate = {2023-02-28},
 volume = {2},
 year = {1954}
}

@InProceedings{KT25,
  author =	{Klein, Nathan and Taziki, Mehrshad},
  title =	{{Dual Charging for Half-Integral TSP}},
  booktitle =	{Approximation, Randomization, and Combinatorial Optimization. Algorithms and Techniques (APPROX/RANDOM 2025)},
  pages =	{21:1--21:22},
  series =	{Leibniz International Proceedings in Informatics (LIPIcs)},
  ISBN =	{978-3-95977-397-3},
  ISSN =	{1868-8969},
  year =	{2025},
  volume =	{353},
  editor =	{Ene, Alina and Chattopadhyay, Eshan},
  publisher =	{Schloss Dagstuhl -- Leibniz-Zentrum f{\"u}r Informatik},
  address =	{Dagstuhl, Germany},
  URL =		{https://drops.dagstuhl.de/entities/document/10.4230/LIPIcs.APPROX/RANDOM.2025.21},
  URN =		{urn:nbn:de:0030-drops-243879},
  doi =		{10.4230/LIPIcs.APPROX/RANDOM.2025.21}
}
\end{document}